\documentclass{llncs}

\usepackage[latin1]{inputenc}
\usepackage{subfigure}

\usepackage{tikz}
\usepackage{xspace}
\usepackage{amsfonts,amssymb}
\usepackage{amsmath}
\usepackage{lineno}
\usepackage{xcolor}

\newcommand{\AND}{{\sc And}\xspace}
\newcommand{\SET}{{\sc And}}
\newcommand{\point}{p}

\begin{document}

\title{($c$-)\AND: A new graph model}

\author{Mauricio Soto\inst{1}\thanks{Email: \email{mausoto@dim.uchile.cl}. Gratefully   acknowledges   the
    support of  CONICYT, Apoyo al  retorno de investigadores  desde el
    extranjero grant 82130059.}
  \and Christopher Thraves Caro\inst{2}\thanks{Email: \email{cbthraves@gsyc.es}. Christopher Thraves Caro is supported by Spanish MICINN grant Juan de la Cierva, Comunidad de Madrid grant S2009TIC-1692
and Spanish MICINN grant TIN2008--06735-C02-01.}
}
  
  \institute{DIM, Universidad de Chile, Chile
\and
GSyC, Universidad Rey Juan Carlos, Spain}
\maketitle
\begin{abstract}
In this document, we study the scope of the following graph model: 
  each vertex is assigned to a box in a metric space and to a
  representative element that belongs to that box. 
Two vertices are connected by an
  edge if and only if its respective 
  boxes contain the opposite representative element. 
We focus our study on the case where boxes (and therefore representative elements) associated to vertices are 
  spread in the Euclidean line.
 We give both, a combinatorial and an intersection characterization of
   the model.  Based on these characterizations, we determine graph
   families that contain the model (e. g., boxicity $2$ graphs) 
   and
   others  that the  new model  contains 
   (e.  g., rooted  directed path). 
We also study the particular case where each representative element is the center of its
respective box.  In this particular case, we provide constructive representations
  for interval, block and outerplanar graphs. 
 Finally, we show that the general and the particular model are not equivalent by
  constructing  a graph  family  that separates the two cases.  
\end{abstract}
\section{Introduction}\label{sec:intro}
A \emph{disk graph} is a graph where the set of 
vertices corresponds to a collection of points that belong to a metric 
space and an edge connects two vertices if and only if their corresponding 
points are at a distance of at most a parameter $r$.
%
An important application of disk graphs is in the area of \emph{sensor networks}. 
Sensor networks are networks formed by sensor nodes, little devices deployed 
in a geographic area with monitor purposes. Sensors communicate with 
each other via a radio channel. Every sensor covers with its radio signal 
a communication area around it and two sensors communicate with each other 
when they are placed 
within each other communication area. In an ideal model, the communication area of a sensor is a circle. 
Therefore, in the same ideal model, if every sensor covers equally sized communication areas, 
the network formed by sensors is a disk graph. That explains why researchers have used disk graphs 
to represent sensor networks, particularly unit disk graphs  \cite{DBLP:journals/talg/Farach-ColtonFM09} 
or some variations \cite{springerlink:10.1007/s11276-007-0045-6}. 

Nevertheless, it is difficult to find such an ideal situation in a real deployment, 
mainly due to physical or geographical restrictions. For instance, when 
the deployment area is irregular, 
the communication area of a  sensor might be shrunken in one direction
due to an obstacle, while, in the opposite direction, the area is free of any obstacle. 
On the other hand, some sensors may have directional antennas which 
produce communication areas that are far from being a circle, 
or that place the sensor location far from the \emph{center}
of its communication area. 
Therefore, the existence of a communication link between two sensors 
is not determined by the distance between them,
neither by the 
intersection of their communication areas. 
In fact, one has to be sure that the
communication areas \emph{cover} the opposite sensor. 

Consequently, we  propose a new graph family that aims to include the different topologies that may be created due to those restrictions. 
Consider a \emph{set} $S$ and an element $\point \in S$ 
as a \emph{representative element} of $S$. Consider now a graph where each vertex corresponds 
to a pair $(S, \point)$ and an edge between two vertices exists if and only if the set associated with a vertex contains 
the representative element of its fellow and vice versa. According to this definition, nonempty 
intersection between two sets is not enough to guarantee the existence of their corresponding edge. 
Moreover, when the sets belong to a metric space, there is no positive distance between two 
representative elements that guarantees the existence of their corresponding edge. 
Therefore, this definition differs from disk graphs, as well as from intersection graphs.

%
In this document, we consider the family induced by the above definition when sets are boxes
in an Euclidean metric space. We aim to understand the properties of such a graph family.  
We focus our study on the case where boxes and representative elements associated to vertices are 
 spread in the Euclidean line. We study the extent of this definition as a graph family. We provide an intersection model and 
 a combinatorial characterization of the model. Additionally, we tackle the subfamily of graphs where all representative elements 
 are the center of its respective boxes.  

%

\section{Definitions}\label{sec:model}
We consider  graphs that  are  finite,  connected,
  undirected,  loopless  and  without  parallel  edges.  For  a  graph
  $G=(V,E)$, 
  we denote  by $V(G)$ and  $E(G)$ the set  of vertices  and edges,
  respectively. When the graph under  consideration is clear, we use only $V$
  and $E$. 
  The edge $\{u,v\}$ is denoted by $uv$. If $uv\in
  E(G)$ we say that $v$ is a \emph{neighbor} of $u$ and vice versa. 
  The    set     of    neighbors    of    $u$     is    denoted    by
  $\mathcal{N}(u)$.  Additionally, the \emph{closed neighborhood} of $u$ is defined as
  $\mathcal{N}[u] :=\mathcal{N}(u)\cup \{u\}$.

A \emph{box} in the $d$-dimensional Euclidean space is the Cartesian product of
  $d$ closed intervals.  A box $B$ is described
  as the set $B=\{(x_1,x_2,\ldots,x_d)\in \mathbb{R}^d: {L}_i\le x_i \le
  {R}_i\}$, where  $L_i$ and  $R_i$ denote the  extreme points  of the
  interval in the $i$-th dimension.
  The \emph{center} of a box is  the Cartesian product of the
  centers of the intervals in each dimension of the box.  Namely, the center of the box 
  $B$ as defined above 
  is the point $((L_1 + R_1)/2,(L_2 + R_2)/2,\ldots,(L_d + R_d)/2)$.
\begin{definition}[\AND-realization]
  An \emph{\AND-realization} of a graph $G$ in the
  $d$-dimensional Euclidean space is a 
  collection of  pairs $\{(B_v,p_v):v\in V(G)\}$  where each vertex
  $v$ is associated to a $d$-dimensional box 
  $B_v$ and to a representative element $p_v \in B_v$, 
  such that: 
  $$
  uv\in E(G)  \Leftrightarrow (\point_v \in B_u)  \wedge (\point_u \in
  B_v).
  $$
  A \emph{central} \AND-realization or $c$-\AND-realization
  of a graph is an \AND-realization in which each representative element $p_v$ is the center of its box $B_v$.
\end{definition}

  We denote by \SET$(d)$ the set of graphs that admit an \AND-realization in the
  $d$-dimensional Euclidean space.  
  The subset of \SET$(d)$ that contains the graphs that admit a
  $c$-\SET-realization in the
  $d$-dimensional Euclidean space is denoted by $c$-\SET$(d)$.
  For simplicity, all along  this document, we use notation ($c$-)\AND
  when we  say something  that concerns to  both classes  $c$-\AND and
  \AND.  

We mainly study 
 sets \SET$(1)$ and $c$-\SET$(1)$. In this context, 
  a box $B_u$ becomes an interval in the Euclidean line that we denote
  by its extreme points $[L(u), R(u)]$. 
   Any ($c$-)\SET$(1)$-realization
   can be modified so that it maintains the graph it represents.  
   For a given realization $\mathcal{R} =\{ ([L(u),R(u)],\point_u)\}_{u \in V(G)}$, 
   we define \emph{$\delta$-translation} and \emph{$\sigma$-scaling} of $\mathcal{R}$ as the realizations
     $\{( [L(u)+\delta,R(u)+\delta], \point_u+\delta)\}_{u\in V(G)}$ and
   $\{( [\sigma\cdot  L(u),\sigma\cdot R(u)],\sigma\cdot\point_u)\}_{u
     \in V(G)}$, respectively.  
%
%

Any  ($c$-)\SET(1)-realization  of  a  graph  induces  a
  natural ordering of its vertices following its
  representative elements, i.e, $v<u$ according to a ($c$-)\SET(1)-realization if and only if 
  $\point_v < \point_u$ in that ($c$-)\SET(1)-realization.  
  In order to properly define this order, each representative element must 
  differ from each other. Nevertheless, it is easy to see that any ($c$-)\SET(1)-realization can be
  modified to fulfill this property.

Given an ordering $\pi$ of the vertices of a graph $G$, we denote by 
  $<_\pi$ the total  order induced by $\pi$. That is,  $u<_\pi v$ if $u$
  appears before  $v$ in $\pi$.
  The \emph{extreme}  vertices of an  order $<_\pi$ are  the vertices
  placed at the first and last position according to $<_\pi$.
  Given a vertex $u$, we denote by $\ell_{\pi}(u)$ and $\rho_{\pi}(u)$
  the leftmost and rightmost neighbors of $u$ in the order, 
  i.e., 
  $\ell_\pi(u)=\{ v \in \mathcal{N}[u] : v<_{\pi}w \: \forall \: w \in \mathcal{N}[u], w \neq v\}$
  and
  $\rho_\pi(u)=\{ v \in \mathcal{N}[u] : w<_{\pi}v \: \forall \: w \in \mathcal{N}[u], w\neq v\}$.

\section{Related work and our contributions}\label{sec:related}
  We compare  the introduced \AND family of graphs with other graph classes. Therefore, we refer the
  reader to an excellent survey authored by Brandst{\"a}dt 
  et  al.~\cite{brandstdt1999graph}  that contains  a description  of
  almost all graph families involved in this document.
This  survey  also  presents containment  relations  between  classes,
  graphs that separate one class form another, and priceless information
  in this area. 
A second  excellent book that  we refer to  the reader is  authored by
  J. Spinrad~\cite{spinrad2003efficient}. This book deals with efficient
  graph representation. For several graph classes, this book considers 
  such questions as existence  of good representations, algorithms for
  finding representations,  questions of characterization in  terms of
  representation, and how the representation affects the complexity of
  optimization problems.  

There exists a vast amount of interesting literature
  related with the graph families that we mention in this document  \cite{mckee1999topics,golumbic2004algorithmic,roberts69,golumbic2004tolerance}: 
 \emph{geometric}, \emph{outerplanar}, \emph{interval}, \emph{max-tolerance} and \emph{boxicity $2$}
  graphs.  
  The  study of  intersection  graphs  dates back  a  long way.  For
  instance,  the  fact  that  all  graph  can  be  represented  as  an
  intersection   graph    was   proved    by   Marczewski    and   Sur
  in~\cite{SzpilrajnSur1945} and by    Erd\"os    et al. in~\cite{Erdos64}.
  Related with particular graph  families, the notion of \emph{boxicity}
of a graph is introduced in \cite{Roberts1969b}. 
On the other hand,  the notion of \emph{book embedding} of a graph is introduced in \cite{Bernhart1979320}
where the authors present some first properties and relations  with  other  invariants  such  as thickness,  genus, 
and  chromatic  number.
 An intersection model for max-tolerance is introduced in \cite{kaufmann06}. Such a model was of great utility at the moment of determining the NP-Hardness of recognition problem for max-tolerance graphs. 
Finally, the book \cite{golumbic2004tolerance} surveys results related
with (\emph{max}-)\emph{tolerance} graphs. 
  

The \SET$(1)$ family has been addressed very recently in parallel to our
work via totally independent way. 
T. Hixon  in his  Master thesis \cite{hixon13}  studies the  family of
\emph{cyclic  segment   graphs}.  This   family  corresponds   to  the
intersection graphs of segments that
lie  on  lines  tangent  to  a   parabola  and  no  two  segments  are
parallel. In his thesis, Hixon also works on the subclass called \emph{hook
  graphs}, in which all segment in the representation need to be tangent to the
parabola. 
The author proves that 
a graph is a hook graph if and only if it is the intersection graph of
a set of axis  aligned rectangles in the plane such  that the top left
corner of each rectangle lies on a unique point on the diagonal. Such 
result  is  equivalent to  the particular  case $d=1$  of Theorem
\ref{prop:boxicity2} in this document.  
The author also proves a combinatorial characterization
for    hook    graphs   which    is    is    equivalent   to    Theorem
\ref{thm:four_poitAND} in our work. 
This result  has been  also obtained independently by Feuilloley \cite{LF13}  in the
study of the gap between Minimum Hitting Size problem and Maximum Independent
Set problem for the \SET(1) family.
Based   on  this   characterization   Hixon   proves  that   interval,
outerplanar, and 2-directional  orthogonal ray graphs (\textsc{2DORG})
are all hook graphs.  
We extend these results by proving that interval and outerplanar
graphs are  not only \SET$(1)$  but also $c$-\SET$(1)$ graphs.  On the
other  hand,  we  extend  the  fact that  interval  graphs  belong  to
\SET$(1)$  by  proving that  a  larger  family, rooted  directed  path
graphs, belongs to \SET$(1)$.  
Moreover,  T. Hixon  proves  that  in a  hook  graph two  non-adjacent
vertices  cannot  be connected  by  three induced  disjoint paths  of
length larger than $4$.  
We prove indeed that such paths, in the case of $c$-\SET$(1)$, cannot
be all longer than $3$. 
The author  gives also  polynomial algorithms  for the
Weighted Maximum  Clique problem and Weighted  Maximum Independent Set
problem; and approximations for the Chromatic Number and 
Clique Covering  Number.  Finally,  according to Hixon,  the \SET$(1)$
family is 
addressed  by  Cantanzaro et  al. \cite{CCH13}  in  the context  of  DNA
sequences. Nevertheless,  as far as we  know, their work has  not been
published. Hence,  it has been  impossible for us compare  our results
with theirs.  

%



\subsection{Our contributions}
The main contribution of this  document is the study of the two graph
families: \SET$(d)$ and $c$-\SET$(d)$.
We study the one-dimensional  version ($c$-)\SET$(1)$ of the families in
which the position of the representative elements induce an order of the vertices. 
\begin{itemize}
\item
We give a characterization of both \SET$(d)$ and
 $c$-\SET$(d)$ families via an intersection model in Subsection \ref{sec:intersection}.
\item
We give a characterization of the \SET$(1)$ family via a combinatorial characterization of the
possible orders of its vertices in any \SET$(1)$-realization in Subsection \ref{sec:comb}.
\item
We give a construction of \SET$(1)$-realizations for \textsc{Interval bigraphs} and \textsc{Rooted directed path} graphs, and 
$c$-\SET$(1)$-realizations for \textsc{Interval}, \textsc{Block} and 
\textsc{Outerplanar} graphs in Subsection \ref{sec:comb} and Section \ref{sec:contention-containment}, respectively. 

\item
Finally, in Section \ref{sec:difference}, we show differences between families \SET$(1)$ and
 $c$-\SET$(1)$,  proving  that  in  the first  case  two  non-adjacent
 vertices  cannot  be  connected   via  three  disjoint  paths  with
 edge-length strictly larger than $3$. While, in the second case, 
two  non-adjacent vertices  cannot  be connected  via three  disjoint
paths with edge-length strictly larger than $2$.
\end{itemize}

\section{Characterizations for \AND graphs}
\label{sec:order}

%
In this section, we show that \SET$(d)$ graphs can be represented by an intersection model. Besides, we give a combinatorial
characterization   for   the   set   of   graphs   that   admit   an
\SET$(1)$-realization. 
From these  characterizations, we obtain containment  relations with
other non trivial graph classes.

\subsection{Intersection graph characterization for \AND graphs}
\label{sec:intersection}

We show in  this subsection that graphs in the  \SET$(d)$ class can be
  represented   as   the   intersection   graph  of   boxes   in   the
  $2\cdot d$-dimensional Euclidean space.

  \begin{theorem}\label{prop:boxicity2}
    A graph $G$ belongs to \SET$(d)$ if
    and only if $G$ is the intersection graph of boxes in the 
    $2\cdot d$ Euclidean space, and each box can be described as 
    $\times_{i=1}^d( [p_i,R_i] \times [-p_i,-L_i] )$
    with $p_i,R_i,L_i>0$ for all $i\in \{1,..,d\}$.
  \end{theorem}   
  \begin{proof}
 Let $G$ be a graph that belongs to \SET($d$). Consider a
   realization 
   \mbox{$\{(B_v,\point_{v})\}_{v\in V}$} of $G$, where
   $B_v=\times_{i=1}^d[L_i(v),R_i(v)]$     and     $p_v=\times_{i=1}^d\,
   p_{(v,i)}$.  
W.l.o.g, assume $L_i(v) >0$ for all $v\in V$ and $i\in\{1,..,d\}$.  
   For each $v \in V$, we define the box $B'_{(v,i)}$  in the
   Euclidean     plane      as     $[\point_{(v,i)},     R_i(v)]\times
   [-\point_{(v,i)},-L_i(v)]$. 
   Finally,  let  $S_v$  be the   $2\cdot d$-dimensional  box  defined  as  the
   cartesian  product of  boxes  $B'_{(v,i)}$, that  is $S_v=  \times_{i=1}^d
   ([p_{(v,i)},R_i(v)] \times  [-p_{(v,i)},-L_i(v)]) $.  

   Consider two vertices  $u, v \in V$, then $uv\in  E$ if and only
   if $p_{(u,i)}\in [L_i(v),R_i(v)]$ and $p_{(v,i)}\in [L_i(u),R_i(u)]$
   for  all   $i\in  \{1,..,d\}$.   For  a   fixed  $i$,   let us assume
   $p_{(u,i)}<p_{(v,i)}$. 
   Thus,                                       $p_{(v,i)}-p_{(u,i)}\le
   \min\{R_i(u)-p_{(u,i)},p_{(v,i)}-L_i(v)\}$ 
   or equivalently  $B'_{(u,i)}\cap B'_{(v,i)}\ne  \emptyset$.  Hence,
   vertices $u$ and $v$ are adjacent if and only if 
   $S_u\cap S_v\ne \emptyset$.
 \end{proof}


\begin{figure}[t]
  \begin{center}
     \resizebox{.85\linewidth}{!}{\begin{tikzpicture}[node distance=2cm,
  thick,main node/.style={circle,draw}, font=\Large]
	
\tikzstyle{dot}=[color=black!50,thick, dotted];

\begin{scope}[thick]
  \node[main node] (1)  at (0,0) {1};
  \node[main node] (3) at  (1.25,-2.5) {3};
  \node[main node] (2) at (2.5,0) {2}; 
  \node[main node] (4) at  (1.25,-5) {4};
  \path
    (1) edge (2)
           edge (3)
    (2) edge (3)
    (3) edge (4);
\end{scope}
\begin{scope}[xshift=5cm,thick]
	\draw[dot]  (0,0) -- (5.5,-5.5);

	\def\s{1.6}
	\def\d{-0.4}
	
	\foreach \i/\c/\l/\r in {1/2/1/2.5,2/3/2.5/1.75,
													3/4/2.5/1.5,4/5/2.75/1.5} 
	{
	\draw[very thick]  (\c-\l,\i/\s+\d)  --++(\l+\r,0); 
	\draw[very thick]  (\c-\l,\i/\s-0.1+\d) -- (\c-\l,\i/\s+0.1+\d);
	\draw[very thick]  (\c+\r,\i/\s-0.1+\d) -- (\c+\r,\i/\s+0.1+\d);
	\draw (\c,\i/\s+\d) node {$\bullet$};	
	\draw (\c,\i/\s+\d) node[above] {$B_\i$};
	\draw [dashed, thick=normal] (\c,\i/\s+\d) -- (\c,-\c+\l);
	\draw (\c,- \c) rectangle ++(\r,\l) ; 
	}
	\foreach \i in {1,2,3,4}
	\draw (\i+1,-\i-1) node[left] {$(p_\i,-p_\i)$};
\end{scope}

\begin{scope}[xshift=11cm,thick]
	\foreach \i/\c/\l/\r in {1/2/1/2.5,2/3/2.5/1.75,
													3/4/2.5/1.5,4/5/2.75/1.5} 
	{
	\draw (\c,- \c) --++(\r,0)--++(-\r,\l)-- (\c,- \c);
	\draw[dot] (\c,-\c+\l)--++(\r,0)-- ++(0,-\l);
	}
\end{scope}
\end{tikzpicture}}
  \end{center}
     \caption{An example of a graph in \SET$(1)$ (left) with  its
       intersection model with boxes (center) and triangles (right hand side). }
     \label{fig:boxicity2}
\end{figure}
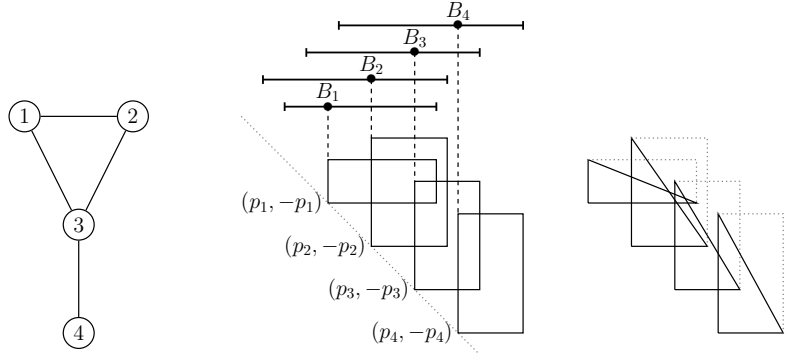

  For the one-dimensional case ($d=1$), Theorem \ref{prop:boxicity2} states that
  \SET$(1)$ graphs  correspond exactly  to the intersection  graphs of
  boxes in the Euclidean plane with its left-lower corner lying on the
  diagonal  $\mathcal{L}:  x+y=0$  (an example is shown in Figure~\ref{fig:boxicity2}).  
  Kaufmann et al.\ in \cite{kaufmann06} proved that  
  max-tolerance graphs correspond to the class of
  intersection of isosceles, axis parallel,
  right triangles (or lower halves of a square). 
  A  different representation of \SET$(1)$ graphs can be obtained
  by  keeping  the  left  lower  half of  the boxes  in  the  intersection
  model (an example is shown in Figure~\ref{fig:boxicity2}). 
  Particularly, when this intersection model is applied to $c$-\SET$(1)$ graphs, we obtain 
  an intersection model of isosceles, axis parallel,
  right triangles (or lower halves of a square).
Therefore, the following corollary holds.
 
  \begin{corollary}
    $c$-\SET$(1)$ $\subset$ \textsc{Max-tolerance}.
  \end{corollary}
  
 
  


\subsection{A combinatorial characterization for the \SET$(1)$ graphs}
\label{sec:comb}

We recall that any \SET$(1)$-realization of a
graph induces a natural ordering of its
vertices by considering their respective representative 
elements.
This ordering needs to have different representative elements 
in order to be totally defined.  
Nevertheless, it is easy to  see that any ($c$-)\SET$(1)$-realization can be
modified to fulfill this property. 

\begin{definition}\label{def:r-order}($\mathcal{R}$-order)
  Given  a  graph $G $ that belongs to \SET$(1)$  and  an  \SET$(1)$-realization  
  $\mathcal{R}$ of $G$ such that all representative elements are different. 
  The \emph{$\mathcal{R}$-order}  of the set $V$, denoted by $<_{\mathcal{R}}$, is
  the total order induced by the representative elements. That is, for any pair of vertices $u$ and $v$: 
  $u<_{\mathcal{R}}v \Leftrightarrow  \point_u<
  \point_v.$
  %
\end{definition}
%
%

Consider  an $\mathcal{R}$-order of a graph $G$ and two vertices 
$u<_\mathcal{R} v$ in $V$. If vertex $u$ has a neighbor $y$ after $v$ ($v<_\mathcal{R} y$) and
$v$ has a neighbor $x$ before $u$ ($x<_\mathcal{R} u$). Then, vertices $u$ and $v$ are
mutually contained in its corresponding intervals. Thus, vertices $u$ and $v$ must be
connected. Indeed, this property characterizes graphs that belong to the set \SET$(1)$. 
Therefore, we introduce the following definition for any ordering of the set of vertices of a graph.  

\begin{definition}\label{def:fourpointcondition}
  Given a graph $G=(V,E)$ and an order $<_\pi$ of its set of vertices. 
  We say that $<_\pi$ satisfies the \emph{four point condition} for \SET$(1)$
  if  and only  if  for  every quadruplet  of  vertices $x,u,v,y$,  it
  holds:

  $$\mbox{If } \: x <_{\pi} u  <_{\pi} v  <_{\pi} y \mbox{ and } xv, uy \in E 
  \Rightarrow uv \in E.$$
  
  Figure~\ref{fig:point_condition} shows a graphic representation of the four point condition for \SET$(1)$. 
\end{definition} 
%
We prove that for any graph $G$  the existence of an ordering of its
set  of  vertices  that  satisfies  the  four  point  condition  for
\SET$(1)$ 
is necessary and sufficient to decide if $G$ belongs to \SET$(1)$.

\begin{figure}[t!]
  \centering
  \resizebox{0.9\linewidth}{!}{\begin{tikzpicture}[node distance=1.5cm]
\begin{scope}
     \node[label=below:$x$](A){$\bullet$};
     \node[right of= A,label=below:$u$](B){$\bullet$};
     \node[right of =B,label=below:$v$](C){$\bullet$};
     \node[right of =C,label=below:$y$](D){$\bullet$};
     \tikzset{style = {bend left}}
     \draw (A.center) edge[out=45,in=135]
 (C.center);
     \draw(B.center) edge[out=45,in=135] (D.center);

\end{scope}

\begin{scope}[xshift=5.3cm]
	 \draw  node {$\Rightarrow$};
\end{scope}

\begin{scope}[xshift=6cm]
     \node[label=below:$x$](A){$\bullet$};
     \node[right of= A,label=below:$u$](B){$\bullet$};
     \node[right of =B,label=below:$v$](C){$\bullet$};
     \node[right of =C,label=below:$y$](D){$\bullet$};
     \tikzset{style = {bend left}}
     \draw (A.center) edge[out=45,in=135]
 (C.center);
     \draw(B.center) edge[out=45,in=135] (D.center);
     \draw(B.center) edge (C.center);

\end{scope}


  \end{tikzpicture}}
  \caption{Graphic representation of the four point condition for \SET$(1)$.} 
  \label{fig:point_condition}
\end{figure}
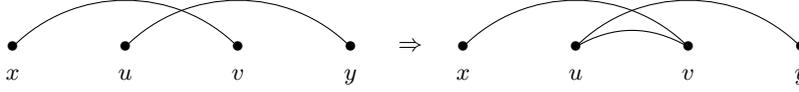

\begin{theorem}\label{thm:four_poitAND}
  A graph $G$ belongs to \SET$(1)$ if  and only if there exists an ordering of its set of
  vertices that satisfies the four point condition for \SET$(1)$.
\end{theorem}
\begin{proof}
  As we have seen previously, the four point condition is necessary  for  any
  \SET$(1)$-realization of $G$. For  the converse, let $<_\pi$ be
  any ordering of the vertices of $G$ which
  satisfies the four point condition. 

  Let $\mathcal{R}_\pi$ be a realization constructed in the following way: 
  representative elements $\point_v$ are embedded in the Euclidean line arbitrarily but respecting the order $<_\pi$. 
  For  each $v \in  V$, we define $B_v$  as the
  interval covering from the  leftmost to the rightmost neighbors of
  $v$        according       to        $<_\pi$,        that       is
  $B_v=[\ell_{\pi}(v),\rho_{\pi}(v)]$.  

  In order to verify that $\mathcal{R}_\pi$ is an \SET$(1)$-realization of  $G$, 
  consider an edge $uv\in E$ with $u<_\pi v$. By definition
  of $\mathcal{R}_\pi$, it holds that $u\in B_v$ and $v\in B_u$. On the other hand,
  if $u\in B_v$ and $v\in  B_u$, then there exist vertices $y\in \mathcal{N}(u)$ and
  $x\in \mathcal{N}(v)$ such that $x<_\pi u<_\pi v<_\pi y$. Thus, vertices $u$
  and $v$ are neighbors by the four point condition.
\end{proof}

\begin{remark}
  Note that the above construction allows us to place the representative elements of the vertices in the integers ranging from $1$ to $n$.
  Hence, any \SET$(1)$ graph can be represented as the collection of $B_v=\{[\ell_{\pi}(v),\rho_{\pi}(v)], \point_v\}$ for all $v \in V$, 
  where $\ell_{\pi}(v),\rho_{\pi}(v)$ and $\point_v$ are integers ranging from $1$ to $n$. On the other hand, in any \SET$(1)$-realization, 
  adjacency between two vertices $u$ and $v$ can be tested by performing four operations in order to check $\point_u \in [\ell_{\pi}(v),\rho_{\pi}(v)]$
  and $\point_v \in [\ell_{\pi}(u),\rho_{\pi}(u)]$. Therefore, we can conclude that the family of graphs \SET$(1)$ admits an \emph{implicit representation} as 
  defined in \cite{spinrad2003efficient}, which is: an \emph{implicit representation} of a graph $G$ is defined as a representation of $G$ 
  that assigns $O(\log n)$ bits to each vertex, such that there is an adjacency testing algorithm that decides adjacency between two vertices $u$ and $v$
  based only on the bits stored at vertices $u$ and $v$.  
\end{remark}

The four point condition is a useful tool to recognize 
graph families that belong to the set \SET$(1)$ as well as families that do
not belong to it.  
We  now present  three graph  families that  belong to  \SET$(1)$, for
which we  show the existence  of an  ordering that satisfies  the four
point condition.  

A graph is a \emph{rooted directed path graph} (also known as directed path
graphs) if it has an intersection
model consisting of directed paths in a rooted directed tree, where
every     arc    is     oriented    from     the    root     to    the
leaves. Figure~\ref{fig:rooted_tree} shows an example of a rooted path
tree and its corresponding intersection model. 

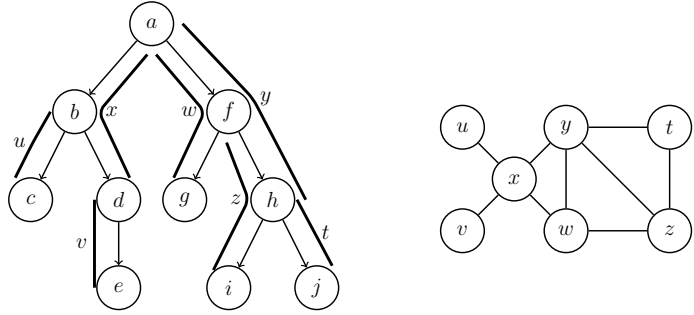
\begin{figure}[t]
  \centering
    \resizebox{.77\linewidth}{!}{

\tikzstyle{level 1}=[level distance=2cm, sibling distance=3.5cm]
\tikzstyle{level 2}=[level distance=2cm, sibling distance=2cm]

\tikzstyle{bag} = [draw, circle, thick,minimum size=24pt,text centered]

\begin{tikzpicture}[grow=down, sloped, edge from parent/.style={draw,->=stealth, thick}, font=\large]

\begin{scope}[scale=.85]
\node[bag](A) {$a$}
    child {
        node[bag] (B) {$b$}        
            child {
                node[bag] (C) {$c$}
            }
            child {
                node[bag] (D) {$d$}
                child {
                      node[bag](E) {$e$} 
                }
           }
    }
    child {
        node[bag] (F) {$f$}
        child {
                node[bag] (G) {$g$}
            }
            child {
                node[bag] (H) {$h$}
			       child {
                        node[bag] (I) {$i$}
                    }
                   child {
                        node[bag] (J) {$j$}
                   }
            }
    };

 \draw[ultra thick, rounded corners] (-.1,-.7)--(-1.2,-2) node[right] {$x$}--(-0.5,-3.5) ;
 \draw[ultra thick,rounded corners](-2.3,-2) -- (-2.7,-2.7) node[left]{$u$}--(-3.1,-3.5);
 \draw[ultra thick,  rounded corners] (.7,0)--(2.3,-1.7) node[right]{$y$}--(3.5,-4);
 \draw[ultra thick] (-1.3,-4)  -- (-1.3, -5) node[left]{$v$} --  (-1.3,-6) ;
 \draw[ultra thick, rounded corners] (.1,-.7)--(1.2,-2) node[left]{$w$}--(0.5,-3.5);
 \draw[ultra thick, rounded corners] (1.7,-2.7)--(2.2,-4) node [left] {$z$}--(1.4,-5.6);
\draw[ultra thick,  rounded corners] (3.3,-4)-- (3.7,-4.75) node [right]  {$t$}  -- (4.1,-5.5);

\end{scope} 

\begin{scope}[xshift=6cm, yshift=-2cm, thick]
  \node[bag] (U)  at (0,0) {$u$};
  \node[bag] (X) at  (1,-1) {$x$};
  \node[bag] (V) at (0,-2) {$v$}; 
  \node[bag] (Y) at  (2,0) {$y$};
  \node[bag] (W) at (2,-2) {$w$};
  \node[bag] (Z) at (4,-2) {$z$};
  \node[bag] (T) at (4,0) {$t$};
  \path
    (X) edge (U)
           edge (V)
		  edge (W)
          edge (Y)
    (W) edge (Y)
		edge (Z)
	(Y) edge (Z)
		edge (T)
	(Z) edge (T);
\end{scope} 

\end{tikzpicture}}
  \caption{At the right hand side, this figure shows an example of a rooted directed path graph. At the left hand side, this figure shows a rooted directed path representation of the graph. An inverse DFS on that tree
   is $(j\,i\,h\,g\,f\,e\,d\,c\,b\,a)$ inducing the following order
    of the vertices  of the graph: $(t\,z\,y\,w\,v\,x\,u)$.}
  \label{fig:rooted_tree}
\end{figure}

\begin{corollary}\label{cor:families-in-and}
  \textsc{Rooted directed path} graphs belong to the set \SET$(1)$. 
\end{corollary}
\begin{proof} 
  In order  to prove  the
  Corollary,  we  give an  ordering  of  the  vertices of  any  rooted
  directed  path graph  $G$ such  that  the four  point condition  for
  \SET$(1)$ is satisfied.  
  Let $G = (V,E)$ be a rooted directed path graph and $T = (K,F)$ be a
  tree with an intersection model  of $G$ consisting of directed paths
  in $T$.  
  For each $v$ in $V$, let us denote by $K_v$ the directed path in $T$
  corresponding to vertex $v \in V$. Note that for every $v$, $K_v$ is
  a subset of $K$.  
  We order $K$, the vertex set  of $T$, using an inverse DFS ordering,
  i.e.,     first     order     $K$     according     to     a     DFS
  (cf. \cite{golumbic2004algorithmic}), and then inverse the ordering.  
  Let  us   denote  by  $\pi:K\rightarrow  \{1,2,\ldots,   |K|\}$  the
  permutation given  by the ordering of  $K$, i.e., $\pi(k) =  i$ when
  $k$ is the $i$-th vertex in the inverse DFS ordering.  
  We define as  well the following notation: $\pi(K_v) =  \{\pi(k) : k
  \in K_v\}$.  
  %
  Since $K$ has an inverse DFS ordering and the fact that $G$ is a rooted directed path graph, for every vertex $v \in V$, $\pi(K_v)$ is increasing when $K_v$ is traversed bottom-up in the tree. 
  %

  Now, vertex  set $V$ is  ordered according  to the minimum  value of
  $\pi(K_v)$. If required, break ties randomly.  
  %
  %
  %
  %
  %
  In order  to conclude the proof  of the Corollary, we  show now that
  the  ordering  of  $V$  satisfies   the  four  point  condition  for
  \SET$(1)$. 
  The proof is by contradiction. Assume that there exist four vertices
  in $V$ such that they violate the condition. I.e., 
  consider  four  vertices $x<u<v<y$  in  $V$  such that  $\{xv,  uy\}
  \subseteq E$ but $uv \notin E$.  
  Since $x<u<v<y$ in the ordering of $V$, it holds $\min \pi(K_x) \leq
  \min \pi(K_u) \leq\min \pi(K_v) \leq  \min \pi(K_y)$ in the ordering
  of $K$.  

  Given that  $xv \in E$, it  holds that $\pi(K_x) \cap  \pi(K_v) \neq
  \emptyset$. Since $x<u<v$, $\min  \pi(K_u) \leq \min \{\pi(K_x) \cap
  \pi(K_v)\}$.  
  Furthermore, due to $uv \notin E$, 
  for every $j \in \pi(K_u)$ it holds that $j < \min \pi(K_v)$. 
  Now,   since   $uy   \in   E$,   hence   $\min   \pi(K_y)   <   \min
  \pi(K_v)$. Therefore, we obtain $y < v$, which is a contradiction.
\end{proof}
\begin{corollary}
  \textsc{Outerplanar} graphs belong to \AND(1).
\end{corollary}
\begin{proof}
  In order to prove that outerplanar graphs belong to \SET$(1)$, let us
  recall the definition of \emph{page embedding} of a graph
  (cf.  \cite{Bernhart1979320}).  A $k$-page embedding, or book
  embedding, of a graph $G$ consists 
  in an linear ordering of the vertices of $G$ which are drawn on a
  line (the  spine of  the book) together  with a partition  of the
  edges into $k$ pages such that  two edges in the same page do not
  cross. The  \emph{pagenumber} of a  graph is the smallest  $k$ for
  which the graph has a $k$-page embedding. In \cite{bilski92},
  Bilski proved that outerplanar graphs are exactly the  graphs with
  pagenumber one. Therefore, for any outerplanar graph there exists
  an ordering of its vertices in which the edges do not cross. Such an
  ordering 
  satisfies the four
  point condition for  \SET$(1)$. 
%
\end{proof}

In contrast to the previous corollary, the four point
condition helps as well to discard a graph from the \SET$(1)$ set. 
\begin{corollary}\label{cor:two-universal-no-and} 
  Let  $G$ be a graph such that all pairs of  vertices $u,v\in V$ have
  at  least two  non adjacent  common neighbors.  Then $G$  does not
  belong to \SET$(1)$.
\end{corollary}
\begin{proof}
  The proof is by contradiction. Let us assume that there exists a graph $G$ that belongs to \SET$(1)$  
  such that all pairs of  vertices $u,v\in V$ have
  at  least two  non adjacent  common neighbors. 
  In order to reach the contradiction, we give four vertices in $V$ that do not satisfy the four point condition for  \SET$(1)$. 
  Let $\mathcal{R}$ be an \SET$(1)$-realization for $G$. Consider the two extreme vertices of $\mathcal{R}$, say vertices $x$ and $y$.  
  There exist two vertices $u$ and $v$ such that $uv \notin E$ and $\{xu, xv, uy, vy\} \subseteq E$. 
  Now, for any order of vertices $u$ and $v$, it holds that the quadruplet $x,u,v,y$ does not satisfy the four point condition for \SET$(1)$.
\end{proof}

\section{Subclasses of  $c$-\SET$(1)$}
\label{sec:contention-containment}

In this section, we establish the relation between
the  $c$-\SET$(1)$ family and other well-known graph families. 
  Particularly, we enhance the result by Hixon \cite{hixon13} 
 by showing that interval and
  outerplanar graphs belong not only to \SET$(1)$ but also to $c$-\SET$(1)$.  
%

   \begin{theorem}\label{thm:interval_cAND}
     The set of \textsc{Interval} graphs is a subset of $c$-\SET$(1)$.
   \end{theorem}
   \begin{proof}     
   Let $G$  be an interval  graph.  In \cite{Olariu199121}  Olariu
   proves that for any interval graph there exists an ordering $<_\pi$ of its
   vertex set $V$ such that
   for all triplet $u,v,w\in V$ with $u<_\pi v<_\pi w$ and $uw\in E$
   then $uv\in E$. Moreover, this order can be obtained in linear time.
   Consider such an ordering for the vertex set $V$. 
   For  the sake  of  simplicity, we  relabel  vertices in $V$ from  $1$ to  $n$
   according to the ordering  $<_\pi$.

   We construct a $c$-\SET$(1)$-realization of $G$ greedily. 
   At   the  $i$-th   step,  we   include  the   vertex  $i$   in  the
   $c$-\SET$(1)$-realization. 
   The inclusion is performed in such a way that, at the end of
   the step $i$, it holds, for all $ j,k,w$ in $\{1,\ldots,i\}$, that:

\begin{minipage}{0.5\linewidth}
       \begin{enumerate}
         \item $p_{k-1}<p_k$\label{item:order}
         \item $\rho(j)<_\pi\rho(k)\Rightarrow R(j) < R(k)$. \label{item:rho}
      \end{enumerate} 
\end{minipage}
\begin{minipage}{0.5\linewidth}
       \begin{enumerate}
         \setcounter{enumi}{2}
         \item $L(j)<p_{\ell(j)}$ \label{item:left}
       \item $\rho(k)<_\pi j\Leftrightarrow R(k)< p_j$. \label{item:right}
      \end{enumerate} 
\end{minipage}
\\[0pt]

  Condition \ref{item:order} ensures that representative elements
   are  placed according to  order $<_\pi$.  Condition \ref{item:rho}
   ensures that right extremes of intervals 
   are in the same order than the values of $\rho(\cdot)$.  
   Finally, conditions \ref{item:left} and \ref{item:right} guarantee 
   that the partial realization at  the end of step $i$ corresponds to
   the subgraph induced by vertices $1,2,\ldots, i$.  Thus, at the end of
   the construction a $c$-\SET$(1)$-realization of $G$ is obtained.

   At the first  step, vertex $1$ is included so that $p_1 = 0$ and $[L(1), R(1)] = [-1,1]$.
   At the end of the first step
   all conditions are satisfied.   Let us suppose that all
   conditions hold at the end of the step
   $i-1$.  
   We include vertex $i$ in the $c$-\SET$(1)$-realization in two phases: 
   \begin{itemize}
   \item First, we set the position of representative element
     $p_{i}$ respecting conditions \ref{item:order} and
     \ref{item:right}.  That is, the representative element 
     is placed after $p_{i-1}$ and it
     is contained only by  intervals associated to its previous 
     neighbors.  
   \item Second,  we set the interval  associated to $i$  such that it
     contains all its previous neighbors, according to
     condition \ref{item:left}. Finally, we
     modify, if necessary, the  interval of previous vertices in order
     to satisfy conditions \ref{item:rho}. 
   \end{itemize}
For the first phase we remark that if two vertices $j,k$ have
  labels   smaller   than   $i$  and   $j\notin\mathcal{N}(i)   \wedge
  k\in\mathcal{N}(i)$ then $\rho(j)<i\le \rho(k)$.
  Therefore,  by condition  \ref{item:rho},  we have  that  $R(j) <  R(k)$.
  Thus, by defining $L=\max\{R(j):\:
  j\notin\mathcal{N}(i)\}$ and 
  $R=\min\{R(k):\:k\in\mathcal{N}(i)\}$, it holds $L<R$.
  Notice that in between $L$ and $R$ there might exist some
  representatives elements.  Hence, by setting $p_{i}$ as $(\max\{p_{i-1},
  L\}+R)/2$, conditions \ref{item:order} and \ref{item:right} hold and
  first phase is concluded. 

In order to set the extremes of interval $B_{i}$, let define  
   $P_i=\{j<i:\: \rho(j)<\rho(i)\}$,  the set  of all  vertices having
   its last neighbor before $i$. We recall that
  condition  \ref{item:rho} imposes  that $R(j)<R(i)$  for  all vertex
  $j$ in $P_i$.  If $R'$ denotes the $\max\{R(j):\:j\in P_i\}$ then 
  it must holds that $R'<R(i)$.  
  On the other hand, the  interval $B_i$ must contain $p_{\ell(i)}$ so
  that condition \ref{item:left} is satisfied.  Then, let define $r_i$ as 
  $\max\{p_i-p_{\ell(i)}, R'-p_i\}+1$. We set $L(i)=p_i-r_i$ and
  $R(i)=p_i+r_i$ so all conditions are satisfied
  for vertices in $P_i$.  However, condition \ref{item:rho},
  does not necessary hold for vertices that do not belong to $P_i\cup \{i\}$.  
To overcome this problem, we extend the
  intervals of those vertices by $2r_i$. That is, we re-define 
  $B_j$  as $[L(j)-r_i,R(j)+r_i]$  for  all  $j\notin P_i\cup  \{i\}$.
  Thus,  since  $R(i)=p_i+r_i<R(j)+r_i$  condition  \ref{item:rho}  is
  satisfied for all vertices in $V$.
\end{proof}

The   rest   of   the   section   aims   to   prove   that
  \textsc{Outerplanar}  graphs belong to  $c$-\SET$(1)$. We  first show
  that cycles belong to $c$-\SET$(1)$.  Moreover, we
  show that any realization of a cycle has a specific structure. 
  Secondly, we construct a procedure to combine biconnected components
  and show how to ``glue'' two different cycles by an edge. 

\begin{lemma}\label{lem:cycle}
  Let  $C_n$  be  a  cycle  of  length  $n$,  then  $C_n$  belongs  to
  $c$-\SET$(1)$.  Furthermore, let $\mathcal{R}$ be 
  an \SET$(1)$-realization  of $C_n$ and $\pi$ be the  permutation induced by
  $<_\mathcal{R}$. Then, there
  exists a clockwise (or anticlockwise) labeling $l:V\rightarrow \{1,2, \ldots, n\}$ such that: \\
  \noindent 1. \label{item:cycleExt}Extreme vertices are adjacent and
  $\pi(l^{-1}(1))=1                                             \wedge
  \pi(l^{-1}(n))=n$.  \\
  2. \label{item:cycleAND}For all $u\in V, \,
  |l(u)-\pi(u)|\le 1$ \\ 
  3. \label{item:cyclecAND}If $\mathcal{R}$ is
  a $c$-\SET$(1)$-realization then for all $u\in V, \,l(u)=\pi(u)$. 

   \end{lemma}
   \begin{proof} 
Let  $C_n$  be a  cycle.  We  prove  that  $C_n$ belongs  to
  $c$-\SET$(1)$ by constructing a realization.  
  Let us label the vertex set $V$ clockwise starting in 
  an arbitrary vertex.  Given $0<\epsilon<1$, we associate 
  to each vertex $i\in \{2,\ldots,n-1\}$ the interval
  $([i-(1+\epsilon),i+(1+\epsilon)]$ and the representative element $p_i=i$. 
  Extreme vertices are 
  assigned to pars (interval, representative element) $([2-n-\epsilon, n+\epsilon], 1)$ and
  $([1-\epsilon, 2n-1+\epsilon],n)$, respectively.  It is easy to check that
  the previous  defined  realization  is  actually  a  $c$-\SET$(1)$-realization  for
  $C_n$. 

 Consider  an \SET$(1)$-realization $\mathcal{R}$  of the cycle
   $C_n$.  If  $n=3$  the  representative  elements are  always  in  a
   (anti-)clockwise order.  Assume then that $n>3$.
   We define a clockwise (or anticlockwise) labeling $l$ of $V$ as follows: 
   (1) the vertex with label 1 has the minimum value of $\point_u$,
   i.e., ($\pi\circ l^{-1}(1)=1$) and, 
   (2) the vertex with label 2 is the neighbor of 1 with the smaller  
   position in the order: $\pi\circ l^{-1}(2)<\pi\circ l^{-1}(n)$.

Condition \ref{item:cycleExt} is proved by contradiction. Note that by definition $l^{-1}(1)$ is an extreme vertex.
Hence, assume
  that $l^{-1}(n)$ is not a extreme vertex. 
   Define  $w$ as follows: $p_{l^{-1}(n)} < p_w$ and $l(w) \leq l(w')$ for all $w'$ such that $p_{l^{-1}(n)} < p_{w'}$. 
  By the definition of the
  labeling,  it holds  that  $l(w)>2$, moreover  $w$  has a  neighbor
  placed between the vertices with labels 1 and $n$, which we
  denote by $v$. 
  We conclude that quadruplet 
  $l^{-1}(1)  <_\mathcal{R} v<_\mathcal{R}  l^{-1}(n) <_\mathcal{R}  w$  violates four
  point condition, which is a contradiction. Hence, vertex  $l^{-1}(n)$ is an extreme vertex. 

We  prove \ref{item:cycleAND}  and  \ref{item:cyclecAND} greedily.  
  First, let us introduce
  some definitions. We say that a vertex $u$
  satisfies the \emph{pre-condition} if for all $v$ such that $p_v < p_u$ it holds $l(v) < \pi(u)$.
  Clearly, extreme vertices
  satisfy the pre-condition. 
  Let $w$ be a vertex that satisfies the pre-condition but such that
  $l(w)\neq \pi(w)$, then it must hold that $l(w)> \pi(w)$. Let us denote by $v$ the vertex such that 
  $p_v < P_w$ and $l(v') \leq l(v)$ for all $v'$ such that $p_{v'} < P_w$.
 By the definition of $w$,
  it holds that $l(v)<n-1$.  We denote by $x$ the neighbor of $v$ with label
  $l(v)+1$,  thus $w<_\mathcal{R}  x$.  Let $w'$  be  the vertex  in
  between $v$ and $x$ with the maximum label. Since $l(v)<l(x)<n$ then $w'$
  must have a neighbor $y$ (with label $l(w')+1$) such that $p_x < p_y$. 
  Thus,  by  the  four point  condition  in  the
  quadruplet  $v <_\mathcal{R}  w'<_\mathcal{R} x<_\mathcal{R}  y$, 
  vertices $w'$ and $x$ must be neighbors.  We conclude that $w=w'$ and
  $l(w)=l(x)+1=l(v)+2$.   Additionally, the  vertex  immediately after
  $x$, that is, in the position $\pi(x)+1$, satisfies the pre-condition. 

As    we    state   before,    extreme    vertices   satisfy    the
  pre-condition.   Let  $w$   be   the  first   vertex  according   to
  $<_\mathcal{R}$ such that $l(w)\neq \pi(w)$. By definition, $w$
  satisfies the pre-condition.  Thus, by the previous discussion, we have that
  $l(w)=\pi(w)+1$. Furthermore, the next  vertex in the ordering, say
  $x$, has label $l(w)-1$ and then $l(w)-\pi(w)=1 \wedge l(x)-\pi(x)=-1$. 
Furthermore, the next vertex in the ordering
  must satisfy the pre-condition. 
  By iterating over vertices according to the order $<_\mathcal{R}$,
  we verify that Condition \ref{item:cycleAND} holds. 
Finally, consider the case when  $\mathcal{R}$ is a
  $c$-\SET$(1)$-realization.    Let   $v
  <_\mathcal{R}  w<_\mathcal{R} x<_\mathcal{R}  y$  be the  quadruplet
  previously constructed, where $l(y)=l(w)+1$. 
If $p_{x}$ is placed in the left half of the interval $[p_v,p_y]$ then 
  $vw\in E$, otherwise $xy \in E$ which yields a contradiction. Thus,
  for all vertices in the $c$-\SET$(1)$-realization $l(w)=\pi(w)$.
   \end{proof}


 \begin{definition}[Safe vertex]
  Let $G$ be a graph in ($c$-)\SET$(1)$.  We say that a vertex $v\in V$ is
  \emph{safe} in $G$ if there exists a ($c$-)\SET$(1)$-realization $\mathcal{R}=\{ (B_u,\point_u)\}_{u
    \in V(G)}$ such that $v\in B_w$ if and only if $v=w \vee vw\in E(G)$.  
\end{definition}
   
   A safe vertex allows the union of two
   different biconnected components. 
   This important property comes from the fact that in a
   realization where  a vertex $v$ is  safe, the interval  $B_v$ can be
   extended as much as required without modifying the original graph.  

   \begin{lemma}\label{lem:glue}
     Consider two graphs $G_1, G_2\in$ ($c$-)\SET$(1)$ and two vertices $w_1\in
     V(G_1)$ and $w_2\in V(G_2)$. 
     Let  $G$  be the  graph  obtained  by  identifying $w_1$  and
     $w_2$. If $w_2$ is safe in $G_2$, then it holds that $G\in$ ($c$-)\SET$(1)$.
   \end{lemma}
  \begin{proof}
    Let $G$  be a  graph obtained by  the identification  of vertices
    $w_1$ and $w_2$ of two different graphs $G_1$ and $G_2$. 
    Consider two ($c$-)\SET$(1)$-realizations 
    $\mathcal{R}_1=\{(B_u,p_u)\}_{u\in V(G_1)}$ and 
    $\mathcal{R}_2=\{(B_u,p_u)\}_{u\in V(G_2)}$ 
    of $G_1$ and $G_2$, respectively, such that $w_2$ is safe in
    $\mathcal{R}_2$.  
   We denote  by $\Delta$ the  minimum distance between $p_{w_1}$  and the
   representative  elements of its  neighbors, that  is $\Delta=\min_{u\in
     \mathcal{N}(w_1)}\{|p_{w_1}-p_u|\}$.
   Let $B$ be an interval
   such that $\cup_{v\in V(G_2)}\{B_v\}\subseteq B$ and
   denote by $L$ its length.  
   We construct the realization $\mathcal{R}'_2=\{(B'_v,p'_v)\}_{v\in V(G_2)}$ 
   from $\mathcal{R}_2$ by the following procedure:
   \begin{itemize}
   \item apply a $(-p_{w_2})$-translation in order to place the representative element of $w_2$ in the origin, 
   \item scale the realization by a factor $\Delta/(2L)$,
   \item perform a $(p_{w_1})$-translation in order to
     equals the position of representatives elements of $w_1$ and $w_2$.
   \end{itemize}
   Let $B_w$  be the  interval with center  in $\point_{w_1}$  and of
   length equal to the maximum between $B_{w_1}$ and $B'_{w_2}$. 
   Then, let us define 
   $\mathcal{R}=\mathcal{R}_1\cup \mathcal{R}'_2 \smallsetminus
   \{(B_{w_1},\point_{w_1}),(B'_{w_2},\point'_{w_2})\}            \cup
   (B_w,\point_{w})$.  
   We see that $\mathcal{R}$ is a ($c$-)\AND(1) realization for $G$. 
   In fact, all edges $uv \in E(G_1)\cup
   E(G_2)$ are induced by $\mathcal{R}$. Furthermore by the
   definition of $\mathcal{R}'_2$ and the fact that  $w_2$ is
   safe, no new edges are generated by $\mathcal{R}$. 
  \end{proof}
   
     Given a graph $G$, the \emph{block tree} of $G$ is the
     graph having two types of vertices: 
     blocks and cut-vertices (cf. \cite{bondy2007graph}). A block vertex represents a maximal
     biconnected component of $G$ while cut-vertices are the 
     articulation points between blocks. The edges of the block tree
     join blocks with cut-vertices. A block is adjacent to a
     cut-vertex if the block contains the cut-vertex. 
     Figure~\ref{fig:block_graph} shows an
     example of a graph and its block tree. 

\begin{figure}[t!]
  \begin{center}
     \resizebox{.7\linewidth}{!}{\begin{tikzpicture}[ thick,main node/.style={circle,draw,fill=black}, font=\Large]
\tikzstyle{block}=[circle,draw]

\begin{scope}[thick,scale=1.75]
  \node[main node]  (1)  at (0,0) {};
  \node[main node] (3) at  (0,1) {};
  \node[main node] (2) at (1.3,1) {}; 
  \node[main node] (4) at  (0.2,2) {};

  \node[main node] (5) at  (2,2) {};
  \node[main node] (6) at  (3,1.2) {};
  \node[main node] (7) at  (2.3,0) {};

  \node[main node] (9) at  (4,1) {};

  \node[main node] (10) at  (4.5,2) {};
  \node[main node] (11) at  (5,1) {};
  \node[main node] (12) at  (4.2,0) {};
  \path
    (1) edge (2)
           edge (3)
    (2) edge (3)
    (2) edge (4)

    (2) edge (5)
    (5) edge (6)
    (6) edge (7)
    (2) edge (7)
    (2) edge (6)

    (6) edge (9)

    (9) edge (10)
    (10) edge (11)
    (12) edge (11)
    (12) edge (9);

\end{scope}

\begin{scope}[xshift=13cm, scale=1.4]
\begin{scope}[thick]
  \node[block]  (1)  at (0,0) {};
  \node[main node] (2) at (1.3,1) {}; 
  \node[block] (4) at  (0.2,2) {};

  \node[block] (6) at  (2.5,1.2) {};
  \node[main node] (9) at  (3.7,1.3) {};
  \node[block] (11) at  (4.8,1.1) {};
  \node[main node] (12) at  (6,1) {};
  \node[block] (10) at  (7,1.1) {};

  \path
    (1) edge (2)
    (2) edge (4)

    (2) edge (6)
    (6) edge (9)
    (9) edge (11)
    (11) edge (12)
    (12) edge (10);
\end{scope}

\end{scope}

\end{tikzpicture}}
  \end{center}
     \caption{This figure shows a graph (left hand side) and its block tree (right hand side). In the block tree representation, white vertices represent 
       maximal  biconnected components, while black  vertices represent
       cut-vertices. } 
     \label{fig:block_graph}
\end{figure}
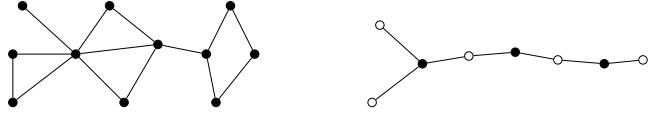

   \begin{theorem}\label{lem:biconnected}
     Let $G$  be a  connected graph  and $T$ be its  block tree.  If all
     maximal biconnected components of $G$ belong to ($c$-)\SET$(1)$ and 
      $T$ can be rooted such that every cut-vertex is safe in 
      its descendants, then $G$ belongs to ($c$-)\SET$(1)$.
   \end{theorem}

   \begin{proof}
     The proof  follows directly  from Lemma \ref{lem:glue}  by adding
      biconnected  components  of  $G$  in  a  breadth-first  traversal
      BFS (cf.~\cite{golumbic2004algorithmic}) order of $T$. 
   \end{proof}
   
The previous result allows us to constructively obtain a realization of a
  graph by  gluing the realization of its  biconnected components. 
  As a consequence, we obtain the following corollary. 
  \begin{corollary}
   \textsc{Block} graphs, graphs in which all biconnected components
     induce a clique, belong to $c$-\SET$(1)$.
\end{corollary}


  An analogous result to Lemma~\ref{lem:glue} can be obtained to identify 
  edges in two different cycles:

  \begin{lemma}\label{lem:glue_cycles}
    Given two cycles $C_n,C_m'$ and two edges $uv \in E(C_n)$ and
    $u'v' \in E(C'_m)$, let $G$ be  the graph obtained by identifying $uv$ and
    $u'v'$. Then, $G\in$ $c$-\SET$(1)$.
  \end{lemma}

  \begin{proof}
   For  $0<\epsilon<1$ we  construct  two $c$-\SET$(1)$-realizations  
   $\mathcal{R}$ and $\mathcal{R'}$ of
   $C_n$ and $C_m'$ respectively according to the procedure described
   in the proof of Theorem~\ref{lem:cycle}. Furthermore, we suppose 
   that $u$ and $v$ are the extreme vertices of realization $\mathcal{R}$
   but $u'$and $v'$ are not the extreme vertices of $\mathcal{R'}$. 
   We perform  a $1/(m-1)$-scaling and a  translation of $\mathcal{R}$
   so that  the positions  of representative elements  of $u$  and $v$
   equal those of $u'$ and $v'$ in $\mathcal{R'}$. 
    The realization $\mathcal{R}\cup \mathcal{R}' \smallsetminus
    \{(B_{u'},u'),(B_{v'},v')\}$ is  a $c$-\SET$(1)$-realization for the
    graph $G$. Notice that this realization can be done with any
    vertex as an extreme (safe) vertex.
  \end{proof}

   \begin{theorem}\label{theo:outerplanar}
     The set of \textsc{Outerplanar} graphs is a subset of $c$-\SET$(1)$.
   \end{theorem}
   \begin{proof}
     Maximal  biconnected  components  of  an  outerplanar  graph  are
     dissections of a convex polygon, which belong to $c$-\SET$(1)$ by
     Lemma~\ref{lem:glue_cycles}.    The  proof   follows   by  gluing
     biconnected components according to Theorem~\ref{lem:biconnected}. 
   \end{proof}

\section{Differences between \SET$(1)$ and $c$-\SET$(1)$}\label{sec:difference}

In  this  section,  we   show  the  difference  between  \SET$(1)$  and
  $c$-\SET$(1)$ via graphs that belong to \SET$(1)$ but
  which does not belong to $c$-\SET$(1)$.  
We start with  a remark upon the fact that the property of being part
  of \SET$(1)$ or $c$-\SET$(1)$ is hereditary, i.e., if a graph $G$
  belongs to ($c$-)\SET$(1)$ then every induced 
  subgraph of $G$ also belongs to ($c$-)\SET$(1)$. 
  Indeed, if  a graph $G$  has a ($c$-)\SET$(1)$-realization  then the
  same  realization is  also a  ($c$-)\SET$(1)$-realization  for every
  induced  subgraph  of  $G$   when  the  corresponding  vertices  are
  deleted. 
From the  hereditary property,  we  define a
  graph $G$ that does not belong to ($c$-)\SET$(1)$ as \emph{minimal} 
  with respect to ($c$-)\SET$(1)$ if and only if every proper induced
  subgraph of $G$ does belong to ($c$-)\SET$(1)$.
 All graphs introduced here that separate 
  \SET$(1)$   and   $c$-\SET$(1)$   are   minimal  with   respect   to
  $c$-\SET$(1)$ and they are based in the following two definitions.

%

 \begin{definition}
Let $H^{l_x,l_y,l_z}$ be a finite graph that consists of two not neighboring vertices, say vertices $a$ and $b$, 
together with three vertex disjoint paths that connect vertex $a$ with vertex $b$.
The three paths that connect vertex $a$ with vertex $b$ follow: path 
$X = \{a=x_0,x_1,x_2,\ldots,x_{l_x-1},x_{l_x}=b\}$, path $Y = \{a=y_0,y_1,y_2,\ldots,y_{l_y-1},y_{l_y}=b\}$ and path $Z = \{a=z_0,z_1,z_2,\ldots,z_{l_z-1},z_{l_z}=b\}$, 
where the edge-length of the paths, denoted by $l_x$, $l_y$, and $l_z$, are larger or equal than $2$.  
A graphic representation of $H^{l_x,l_y,l_z}$ is shown in Figure \ref{fig:hxyz}.
  \begin{figure}[t!]
    \centering
 \resizebox{0.8\linewidth}{!}{\begin{tikzpicture}[node distance=1.2cm, label distance=-4]

\begin{scope}
     \node[label=left:$a$](A){$\bullet$};
     \node[right of= A,label=below:$x_1$](Y1){$\bullet$};
     \node[right of= Y1,label=below:$x_{l_x-1}$](Yl){$\bullet$};
     \node[right of= Yl,label=right:$b$](B){$\bullet$};
     \node[above of= Y1,label=above:$y_1$](X1){$\bullet$};
     \node[right of= X1,label=above:$y_{l_y-1}$](Xl){$\bullet$};
     \node[below of= Y1,label=below:$z_1$](Z1){$\bullet$};
     \node[right of= Z1,label=below:$z_{l_z-1}$](Zl){$\bullet$};

     \node [above of=A,font=\Huge,label=above:$\huge{H^{l_x,l_y,l_z}}$](H1){};

	 \draw (A.center) edge[bend left] (X1.center);
	 \draw (A.center) edge (Y1.center);
	 \draw (A.center) edge[bend right] (Z1.center);

	 \draw (Xl.center) edge[bend left] (B.center);
	 \draw (Yl.center) edge (B.center);
	 \draw (Zl.center) edge[bend right] (B.center);
	
     \draw (X1.center) edge[dashed]  (Xl.center);
     \draw (Y1.center) edge[dashed]  (Yl.center);
     \draw (Z1.center) edge[dashed]  (Zl.center);

  \end{scope}

\begin{scope}[xshift=180]

     \node[label=left:$a$](A){$\bullet$};
     \node[right of= A,label=below:](Y1){};
     \node at (1.75,0) [label=below:$x_{1}$](Yl){$\bullet$};
     \node at (3.5,0) [label=right:$b$](B){$\bullet$};
     \node[above of= Y1,label=above:$y_1$](X1){$\bullet$};
     \node[right of= X1,label=above:$y_{l_y-1}$](Xl){$\bullet$};
     \node[below of= Y1,label=below:$z_1$](Z1){$\bullet$};
     \node[right of= Z1,label=below:$z_{l_z-1}$](Zl){$\bullet$};

     \node [above of =A,font=\Huge,label=above:$\huge{H^{2,l_y,l_z}}$](H1){};

	 \draw (A.center) edge[bend left] (X1.center);
	 \draw (A.center) edge (Y1.center);
	 \draw (A.center) edge[bend right] (Z1.center);

	 \draw (Xl.center) edge[bend left] (B.center);
	 \draw (Yl.center) edge (B.center);
	 \draw (Zl.center) edge[bend right] (B.center);
	
     \draw (X1.center) edge[dashed]  (Xl.center);
     \draw (Y1.center) edge  (Yl.center);
     \draw (Z1.center) edge[dashed]  (Zl.center);	

\end{scope}

\begin{scope}[xshift=360]
     \node[label=left:$a$](A){$\bullet$};
     \node[right of= A,label=below:$x_1$](Y1){$\bullet$};
     \node[right of= Y1,label=below:$x_{2}$](Yl){$\bullet$};
     \node[right of= Yl,label=right:$b$](B){$\bullet$};
     \node[above of= Y1,label=above:$y_1$](X1){$\bullet$};
     \node[right of= X1,label=above:$y_{l_y-1}$](Xl){$\bullet$};
     \node[below of= Y1,label=below:$z_1$](Z1){$\bullet$};
     \node[right of= Z1,label=below:$z_{l_z-1}$](Zl){$\bullet$};

     \node  [above of=A,font=\Huge,label=above:$\huge{H^{3,l_y,l_z}}$](H1){};

	 \draw (A.center) edge[bend left] (X1.center);
	 \draw (A.center) edge (Y1.center);
	 \draw (A.center) edge[bend right] (Z1.center);

	 \draw (Xl.center) edge[bend left] (B.center);
	 \draw (Yl.center) edge (B.center);
	 \draw (Zl.center) edge[bend right] (B.center);
	
     \draw (X1.center) edge[dashed]    (Xl.center);
     \draw (Y1.center) edge  (Yl.center);
     \draw (Z1.center) edge[dashed]  (Zl.center);
\end{scope}

  \end{tikzpicture}}
    \caption{This  figure  shows,  from   left  to  right,  a  graphic
      representation  of   (a  general)  $H^{l_x,l_y,l_z}$,   and  the
      particular cases of $H^{2,l_y,l_z}$ and $H^{3,l_y,l_z}$.} 
    \label{fig:hxyz}
  \end{figure}
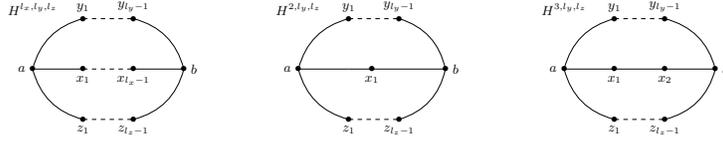
  \end{definition}

  \begin{lemma}\label{lem:h3notand}
 Any  $H^{l_x,l_y,l_z}$ graph such that $l_z \geq l_y \geq l_x > 3$ does not belong to \SET$(1)$.
  \end{lemma}
   \begin{proof}
     The proof is by contradiction. Let $\mathcal{R}$ be an \SET$(1)$-realization of $H^{l_x,l_y,l_z}$. 
   By Lemma~\ref{lem:cycle}, Condition \ref{item:cycleExt}, the extreme vertices of
   the realization must be neighbors. Then, both extremes belong to
   the same path.   Without loss of generality, we  assume that vertex
   $a$ is placed before than vertex $b$ in the realization ($a<_\mathcal{R} b$)
   and that both extremes belong to path $X$, say $x_k$ and
   $x_{k+1}$ with $k \in \{0,\ldots,l_x-1\}$.  Therefore, according to
   Lemma~\ref{lem:cycle}, the  induced cycles  $X\cup Y$ and  $X\cup Z$
   must be oriented clockwise $$(x_k,\ldots,x_0,y_1,\ldots,
   y_{l_y-1},b,\ldots,x_{l_x-1},\ldots,x_{k-1})$$ and anti-clockwise 
   $$(x_k,\ldots,x_0,z_1,\ldots,
   z_{l_y-1},b,\ldots,x_{l_x-1},\ldots,x_{k-1}),$$ respectively. 

On the  other hand, from Condition \ref{item:cycleAND} of Lemma~\ref{lem:cycle}, it
   holds  that  a  labeling  $l$  of an  induced  cycle  satisfies  the
   following  property:  if  for   two  vertices  $u,v$   it  holds
   $l(u)<l(v-1)$ then $\pi(u)<\pi(v)$, i.e., $u<_\mathcal{R}v$.  
Thus,    in   the    cycle   $X\cup    Y$,    $y_2<_\mathcal{R}$   and
   $a<_\mathcal{R}y_j$ for all $j>2$. 
   Symmetrically, for the cycle $X\cup Z$ , $a<_\mathcal{R}z_j$ for all
   $j>2$ and $z_2<_\mathcal{R} b$. 

Finally,  in the  induced  realization  of cycle  $Y\cup  Z$, the  only
  possible  pairs of  extreme vertices  are $(a,y_1),  (y_1,y_2)$ and
  $(a,z_1), (z_1,z_2)$.  If the  extremes are $(a,y_1)$ or $(y_1,y_2)$,
   the cycle $Y\cup Z$ is oriented anti-clockwise and 
  $b<_\mathcal{R}y_2$ which is a contradiction. 
Otherwise, if $(a,z_1)$ or  $(z_1,z_2)$ are the extremes, cycle $Y\cup
  Z$ is oriented clockwise and $b<_\mathcal{R}z_2$ which is also
  a contradiction. 
    \end{proof}

%

We shall see  now that, indeed, the smaller  cases $H^{2,l_y,l_z}$ and
$H^{3,l_y,l_z}$  are  minimal  graphs  that  separate  \SET$(1)$  from
$c$-\SET$(1)$. 
  %
\begin{lemma}\label{lem:threepath-notin-cand}
Any  $H^{l_x,l_y,l_z}$ graph such that  $l_z \geq l_y \geq l_x > 2$ does not belong to $c$-\SET$(1)$. 
\end{lemma}
The proof of this lemma follows the same ideas of the proof of Lemma \ref{lem:h3notand}. 

\begin{proof} 
    The  proof   is  by   contradiction.  Let  $\mathcal{R}$   be  an
    $c$-\SET$(1)$-realization of $H^{l_x,l_y,l_z}$.  
    Paths $X$, $Y$ and $Z$ are defined as the previous proof. 
  Since the extreme vertices of
  the         realization          must         be         neighbors
  (Lemma~\ref{lem:cycle}, Condition \ref{item:cycleExt}),   then  both  extremes
  belong to the same path.  W.l.o.g., we assume that vertex
  $a$ is placed before than vertex $b$ in the realization ($a<_\mathcal{R} b$)
  and that both extremes belong to $X$, says $x_k$ and
  $x_{k+1}$ with $k \in \{0,\ldots,l_x-1\}$.  Therefore, according to
  Lemma~\ref{lem:cycle}, the  induced cycles  $X\cup Y$ and  $X\cup Z$
  must be oriented clockwise and anti-clockwise, respectively. That is:
  $$x_k,\ldots,x_0,y_1,\ldots,
  y_{l_y-1},b,\ldots,x_{l_x-1},\ldots,x_{k-1}$$ and  $$x_k,\ldots,x_0,z_1,\ldots,
  z_{l_y-1},b,\ldots,x_{l_x-1},\ldots,x_{k-1},$$ respectively.  Thus,
  $y_1<_\mathcal{R} b$ and $z_1<_\mathcal{R} b$.

 Consider the  induced  realization  of cycle  $Y\cup  Z$. By  the
 previous discussion, we conclude that $a$ is the left extreme vertex of
 the induced  realization. Thus,  right extreme have  to be  $y_1$ or
 $z_1$. Then, either $b<_\mathcal{R} y_1$ or $b<_\mathcal{R} z_1$ which
 both are contradictions.
  \end{proof}

With the previous Lemma we have presented an infinite family of graphs that do not belong $c$-\SET$(1)$. 
Nevertheless, some of these graphs do belong to \SET$(1)$.

\begin{lemma}\label{lem:threepaths-1-in-and}
Graphs $H^{2,l_y,l_z}$ and $H^{3,l_y,l_z}$ belong to \SET$(1)$ for any $l_y$ and $l_z \geq 2$ and  $l_y$ and $l_z \geq 3$, respectively. 
\end{lemma}
The proof of this lemma follows by giving orderings of the set of vertices of $H^{2,l_y,l_z}$ and $H^{3,l_y,l_z}$ that satisfy the four point condition for 
\SET$(1)$. Figure \ref{fig:proof-h1yz} shows graphically such orders. 

  \begin{proof}
  Consider any $H^{2,l_y,l_z}$ graph. In this case, the first  path has length $2$, therefore, we denote its vertex by $x$ without subindex. 
In order to prove the Lemma, we give an ordering of the vertices of $H^{2,l_y,l_z}$ that satisfies the four point condition for \SET$(1)$. 
Consider the following ordering for $V(H^{2,l_y,l_z})$: $$a,z_1,z_2,\ldots,z_{l_z-1},b, x, y_{l_y-1},y_{l_y-2}\ldots,y_1.$$
For the $i$th vertex, we define $p_i$ to be equal to $i$. 
The intervals are defined  as follow: $I_a = [p_a,p_{y_1}]$; $I_{z_i}
=  [p_{z_i}-1,p_{z_i}+1]$;  $I_b=[p_b,p_{l_y}]$;  $I_x=[p_a,p_x]$;
$I_{y_1}=[p_b,p_{y_1}]$;              $I_{y_i}=[p_{y_i}-1,p_{y_i}+1]$;
$I_{y_1}=[p_a,p_{y_1}]$.

  \begin{figure}[t!]
    \centering
          \resizebox{\linewidth}{!}{\begin{tikzpicture}[node distance=1.2cm, font=\LARGE]
\begin{scope}
     \node[label=below:$a$](A){$\bullet$};
     \node[right of= A,label=below:$z_1$](Z1){$\bullet$};
     \node[right of =Z1,label=below:$z_2$](Z2){$\bullet$};
     \node[right of =Z2,label=below:$z_{l_z-1}$](Zl){$\bullet$};
     \node[right of= Zl,label=below:$b$](B){$\bullet$};
     \node[right of= B,label=below:$x$](X){$\bullet$};
     \node[right of= X,label=below:$y_{l_v-1}$](Yl){$\bullet$};
     \node[right of= Yl,label=below:$y_1$](Y1){$\bullet$};

	 \draw (A.center) edge (Z1.center);
	 \draw (Z1.center) edge (Z2.center);
	 \draw (Z2.center) edge[dashed] (Zl.center);
	 \draw (Zl.center) edge (B.center);
	 \draw (B.center) edge (X.center);
	 \draw (Yl.center) edge[dashed] (Y1.center);	 
	
     \tikzset{style = {bend left}}
     \draw (A.center) edge[out=30,in=150]  (Y1.center);
     \draw (A.center) edge[out=25,in=155]  (X.center);
     \draw (B.center) edge[out=60,in=120]  (Yl.center);
\end{scope}

\begin{scope}[font=\LARGE, xshift=11cm]
     \node[label=below:$y_1$](Y1){$\bullet$};
     \node[right of= Y1,label=below:$x_1$](X1){$\bullet$};
     \node[right of= X1,label=below:$a$](A){$\bullet$};
     \node[right of= A,label=below:$z_1$](Z1){$\bullet$};
     \node[right of =Z1,label=below:$z_{l_z-1}$](Zl){$\bullet$};
     \node[right of= Zl,label=below:$b$](B){$\bullet$};
     \node[right of= B,label=below:$x_2$](X2){$\bullet$};
     \node[right of= X2,label=below:$y_{l_v-1}$](Yl){$\bullet$};
     \node[right of= Yl,label=below:$y_2$](Y2){$\bullet$};

	 \draw (X1.center) edge (A.center);
	 \draw (A.center) edge (Z1.center);
	 \draw (Z1.center) edge[dashed] (Zl.center);
	 \draw (Zl.center) edge (B.center);
	 \draw (B.center) edge (X2.center);
	 \draw (Yl.center) edge[dashed] (Y2.center);	 
	
     \tikzset{style = {bend left}}
     \draw (Y1.center) edge[out=30,in=150]  (A.center);
     \draw (X1.center) edge[out=25,in=155]  (X2.center);
     \draw (B.center) edge[out=45,in=135]  (Yl.center);
    \draw(Y1.center) edge[out=30,in=150] (Y2.center);
\end{scope}

  \end{tikzpicture}}         
\caption{This figure shows a graphic representation of $H^{2,l_y,l_z}$
  and $H^{3,l_y,l_z}$ where vertices are ordered according to the ordering
  given in the proof of Lemma \ref{lem:threepaths-1-in-and}.} 
    \label{fig:proof-h1yz}
  \end{figure}
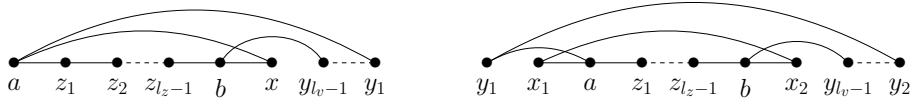  


%

As it can be seen in Figure \ref{fig:proof-h1yz}, according to this ordering of the vertices the only pair of edges 
that crosses one to another are edges $ax$ and $by_{l_y-1}$. Since vertices $b$ and $x$ are connected by an edge in 
$H^{2,l_y,l_z}$, the four point condition is satisfied. Therefore, the graph $H^{2,l_y,l_z}$ belongs to \SET$(1)$.
  
Consider any $H^{3,l_y,l_z}$ graph, we give an ordering of $V(H^{3,l_y,l_z})$ that satisfies the four point condition for \SET$(1)$. 
Consider the following  order for $V(H^{3,l_y,l_z})$: $$y_1,x_1,a,z_1,z_2,\ldots,z_{l_z-1},b,x_2,y_{l_y-1}, y_{l_y-2},\ldots,y_2.$$  
For the $i$th vertex, we define $p_i$ to be equal to $i$. 
The intervals are defined as follow: $I_{y_1} = [p_{y_1},p_{y_2}]$; $I_{x_1}=[p_{x_1},p_{x_2}]$; $I_a=[p_{y_1},p_a+1]$; $I_{z_i} = [p_{z_i}-1,p_{z_i}+1]$; $I_b=[p_b-1,p_{y_{l_y}}]$; $I_{y_i}=[p_{y_i}-1,p_{y_i}+1]$; $I_{y_2}=[p_{y_1},p_{y_2}]$.

%

In order to finish the proof, we have to check that this ordering satisfies the four point condition for \SET$(1)$.
As it can be seen in Figure \ref{fig:proof-h1yz}, there are two pair of edges that crosses one to each other. One pair is composed by edges $y_1a$ and $x_1x_2$. Since vertices $a$ and $x_1$ are neighbors, the condition holds. 
The second pair is composed by edges $x_1x_2$ and $by_{l_y-1}$. Since vertices $x_2$ and $b$ are neighbors, the condition holds. 
Therefore, any graph $H^{3,l_y,l_z}$ belongs to \SET$(1)$.
\end{proof}

We consider important to stress the complete bipartite graph $K_{2,3}$ as a particular case of Lemma \ref{lem:threepath-notin-cand} and Lemma \ref{lem:threepaths-1-in-and}, i.e., 
 $K_{2,3}$ belongs to \SET$(1)$ but it does not belong to $c$-\SET$(1)$. Such an importance comes from the fact that $K_{2,3}$ is 
  the smallest complete bipartite graph that does not belong to $c$-\SET$(1)$.
%
%
%
As a consequence of Lemma \ref{lem:threepath-notin-cand} and the fact that the property of belonging to $c$-\SET$(1)$ is hereditary, 
we can say that any graph that contains a $H^{l_x,l_y,l_z}$ as an induced subgraph 
does not belong to $c$-\SET$(1)$. 
On the other hand, from Lemma \ref{lem:threepaths-1-in-and} 
we know that some of these graphs do belong to \SET$(1)$. 
\section{Future work}\label{sec:future-work}
%
Our results are graphical expressed in Figure \ref{fig:results}.
On the other hand, our work suggests several  directions for future research. In our opinion, the
  most natural  question is  to find a  combinatorial characterization
  for the $c$-\SET$(1)$ family.  Another interesting open
  problem  concerns  to determine  the  complexity  of the recognition
  problem for  both \SET$(1)$ and
 $c$-\SET$(1)$ families.
The  study  of higher  dimensions of  the
  families is an alternative way to continue this research. 
  Another interesting question is the study of the 
  family of graphs   generated  when points  are embedded in a  different
  metric space, for instance the $d$-dimensional torus. 
   
   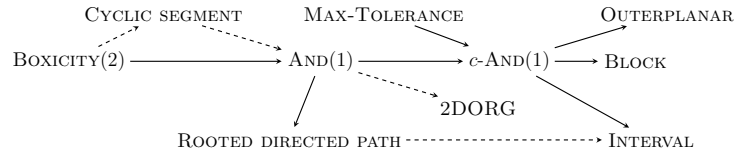
\begin{figure}[t!]
  \begin{center}
     \resizebox{0.8\linewidth}{!}{\begin{tikzpicture}[scale=.8,
  thick,main node/.style={rectangle,rounded corners=.8ex}, font=\Large, arr/.style={>=stealth,->, shorten >=1pt }]
	
\tikzstyle{dot}=[color=black!50,thick, dotted];

\begin{scope}[thick]
	\node[main node] (0)  at (-10,0) {\textsc{Boxicity}(2)};
   \node[main node] (1)  at (-2,0) {\textsc{And}(1)};
  \node[main node] (2)  at (4,0) {$c$-\textsc{And}(1)}; 
  \node[main node] (3)  at (8.5,-2.5) {\textsc{Interval}};  
  \node[main node] (4)  at (9,1.5) {\textsc{Outerplanar}};  
  \node[main node] (5)  at (8,0) {\textsc{Block}};  
  \node[main node] (6)  at (0,1.5) {\textsc{Max-Tolerance}};
  \node[main node] (7)  at (-3,-2.5) {\textsc{Rooted directed path}};
  \node[main node] (9) at (3, -1.5) {\textsc{2DORG}};
  \node[main node] (10) at (-7,1.5) {\textsc{Cyclic segment}};

  \draw
    (0) edge[arr,style=dashed] (10)
    (1) edge[arr] (2)
    (2) edge[arr] (3)
    (2) edge[arr] (4)
    (2) edge[arr] (5)
    (6) edge[arr] (2)
    (1) edge[arr] (7)
    (1) edge[arr,style=dashed] (9) 
    (7)edge[arr,style=dashed] (3)
    (10) edge[arr,style=dashed] (1)
    (0) edge [arr] (1)
;

\end{scope}

\end{tikzpicture}}
  \end{center}
     \caption{Relation    between   the    graph   classes    in   the
       document. Arrows point from the
       superclass  to the  subclass. Dotted  lines represent  previous
       results and solid lines represent results proved in this document.} 
     \label{fig:results}
\end{figure}


%
\paragraph{Acknowledgements}
 The authors would like to  thank Antonio Fern\'andez Anta and Marcos
 Kiwi  because they  are strongly  involved  in the  origins of  this
 study, even more, they contributed with enlightening talks and ideas. 
%
%
\newcommand{\etalchar}[1]{$^{#1}$}

\end{document}